\documentclass[11pt]{article}
\usepackage{fullpage}
\usepackage{amsfonts}
\usepackage{amsmath}
\usepackage{amssymb}
\usepackage{hyperref}

\def\qed{\hspace*{\fill} $\Box$\par\medskip}

\newtheorem{theorem}{Theorem}

\newtheorem{lemma}[theorem]{Lemma}
\newenvironment{proof} {\noindent{\it Proof. }} {{\qed}}
\newenvironment{proofof}[1]{\noindent{\it Proof of #1. }} {{\qed}}

\newtheorem{proposition}[theorem]{Proposition}

\newcommand{\OR}{\mbox{\sc Or}}

\newcommand\Pm{\mathcal{P}}
\newcommand\Pt{{\widetilde{\mathcal{P}}}}
\newcommand\Q{\mathcal{Q}}
\newcommand\Rm{\mathcal{R}}

\newcommand\C{\mathcal{C}}

\newcommand\LL{\mathrm{L}}
\newcommand\RR{\mathrm{R}}

\newcommand\G{{\widetilde{G}}}

\newcommand\wrow{{\mbox{\sl row}}}

\newcommand\wright{{\mbox{\sl right}}}
\newcommand\wleft{{\mbox{\sl left}}}
\newcommand\class{{\mbox{\sl class}}}

\newcommand\wset{{\mbox{\sl set}}}
\newcommand\wtrim{{\mbox{\sl trim}}}
\newcommand\out{{\mathrm{out}}}

\newcommand\rest{\restriction}

\title{Lower bounds for Combinatorial Algorithms for Boolean Matrix Multiplication\footnote{The research leading to these results has received funding from the European Research Council under the European Union's Seventh Framework Programme (FP/2007-2013)/ERC Grant Agreement no. 616787. The second author was also partially supported by the Center of Excellence CE-ITI under the grant P202/12/G061 of GA \v{C}R. The third author was supported in part by the Simons Foundation under Award 332622.}}

\author{Debarati Das\\
Computer Science Institute of Charles University\\
Malostransk{\'e}  n{\'a}mesti 25, 11800 Praha 1, Czech Republic\\
\texttt{debaratix710@gmail.com}
\and Michal Kouck\'y\\
Computer Science Institute of Charles University\\
Malostransk{\'e}  n{\'a}mesti 25, 11800 Praha 1, Czech Republic\\
\texttt{koucky@iuuk.mff.cuni.cz}
\and Michael Saks\\
Department of Mathematics\\
Rutgers University, Piscataway, NJ, USA\\
\texttt{msaks30@gmail.com}
}

\begin{document}

\maketitle

\begin{abstract}
In this paper we propose models of combinatorial algorithms for the Boolean
Matrix Multiplication (BMM), and prove lower bounds on computing BMM in these models. 
First, we give a relatively relaxed combinatorial model which is an extension of the model by Angluin (1976), 
and we prove that the time required by any algorithm
for the BMM is at least $\Omega(n^3 / 2^{O( \sqrt{ \log n })})$. Subsequently, we propose a more general model capable of simulating the
"Four Russians Algorithm". We prove a lower bound of $\Omega(n^{7/3} / 2^{O(\sqrt{ \log n })})$ for the BMM under this model. 
We use a special class of graphs, called $(r,t)$-graphs, originally discovered by Rusza and Szemeredi (1978),
along with randomization, to construct matrices that are hard instances for our combinatorial models.
\end{abstract}


\section{Introduction}

Boolean matrix multiplication (BMM) is one of the core problems in discrete algorithms, with numerous applications including
 triangle detection in graphs ~\cite{Itai77}, context-free grammar parsing ~\cite{Valiant75}, and transitive closure etc. ~\cite{FischerM71, Furman70, Munro71}.
Boolean matrix multiplication can be naturally interpreted as a path problem in graphs.  Given a layered graph
with three layers $A,B,C$ and edges between layers $A$ and $B$ and between $B$ and $C$,   compute the bipartite graph between $A$ and $C$ in which $a \in A$ and $c \in C$ are joined if and only if they have a common neighbor.
If we identify the bipartite graph between $A$ and $B$ with its $A \times B$ boolean adjacency matrix $\Pm$
and the graph between $B$ and $C$  with its  $B \times C$ boolean adjacency matrix $\Q$ then 
 the desired graph between $A$ and $C$ is just the boolean product $\Pm \times \Q$.

Boolean matrix multiplication is the combinatorial counterpart of integer matrix multiplication.  Both involve the computation of $n^2$ output values, each of which
can be computed in a straightforward way in time $O(n)$ yielding  a  $O(n^3)$ algorithm for both problems.  One of the celebrated classical results in algorithms
is Strassen's discovery ~\cite{Strassen69} that by ordinary matrix multiplication has {\em truly subcubic} algorithms, i.e. algorithms that run in time $O(n^{\omega})$ for some $\omega<3$, which
compute the $n^2$ entries by computing and combining carefully chosen (and highly non-obvious) polynomial functions of the matrix entries.
Subsequent improvements ~\cite{CoppersmithW90, Williams12, Gall14a} have reduced the value of $\omega$.

One of the fascinating aspects of BMM  is that, despite its intrinsic combinatorial nature, the asymptotically fastest algorithm known is obtained by treating the boolean entries as integers and applying  fast integer matrix multiplication.  The intermediate calculations done for this algorithm seemingly have little to do with the combinatorial structure of the underlying bipartite graphs.  There has been considerable interest in developing "combinatorial" algorithms for BMM, that is algorithms where the intermediate computations all have a natural combinatorial interpretation in terms of the original problem.  Such interest is motivated both by intellectual curiosity, and by the fact that the fast integer multiplication algorithms are impractical because the constant factor hidden in $O(\cdot)$ is so large.

The straightforward $n^3$ algorithm has a straightforward combinatorial interpretation: for each pair of vertices $a,c$ check each vertex of $B$ to see whether it is adjacent to both $a$ and $c$.    The so-called Four Russians Algorithm by Arlazarov, Dinic, Kronrod, Faradzhev ~\cite{FourRussian70} solves BMM in $O(n^3/\log^2 (n))$ operations, and was the first combinatorial algorithm for BMM with complexity
$o(n^3)$.  Overt the past 10 years, there have been a sequence of  combinatorial algorithms \cite{BansalW12, Chan15, Yu15}  developed  for BMM, all having complexities of the form $O(n^3/(\log n)^c)$ for increasingly large constants $c$.
The best and most recent of these, due to Yu ~\cite{Yu15} has complexity $\hat{O}(n^3/\log^4 n)$ (where the $\hat{O}$ notation suppresses $poly(\log\log(n))$ factors.   (It should be noted that the algorithm presented in each of these recent papers is  for the problem of determining whether a given graph has a triangle; it was shown in \cite{WilliamsW10} that a (combinatorial) algorithm for triangle finding with complexity $O(n^3/\log^cn)$ can be used as a subroutine to give a (combinatorial) algorithm for BMM with a similar complexity.)


While each of these combinatorial  algorithms uses interesting and non-trivial ideas, each one saves only a polylogarithmic factor as compared to the straightforward algorithm, in contrast with the algebraic algorithms which save a power of $n$. The motivating question for the investigations in this paper is: Is there a truly subcubic combinatorial algorithm for BMM?  We suspect that the answer is no.

In order to consider this question precisely, one needs to first make precise the notion of a combinatorial algorithm.  This itself is challenging.  To formalize the notion of a combinatorial algorithm requires some computation model which specifies what the algorithm states are, what operations can be performed, and what the cost of those operations is.
If one examines each of these algorithms one sees that the common feature is that the intermediate information stored by the algorithm is of one of the following three types (1): for some pair of subsets $(X,Y)$ with $X \subseteq A$ and $Y \subseteq B$,
the submatrix (bipartite subgraph) induced by $\Pm$ on $X \times Y$ has some specified  monotone property (such as, every vertex in $X$ has a neighbor in $Y$),  (2)  for some pair of subsets $(Y,Z)$ with $Y \subseteq B$ and $Z \subseteq C$, the bipartite subgraph induced by $\Q$ on $Y \times Z$ has some specific monotone property, or (3) 
for some pair of subsets $(X,Z)$ with $X \subseteq A$ and $Z \subseteq C$, the bipartite subgraph induced by $\Pm \times \Q$ on $X \times Z$ has some specific monotone property.

If one accepts the above characterization of the possible information stored by the algorithm, we are still left with the problem of specifying the elementary steps
that the algorithm is permitted to make to generate new pieces of information, and what the computational cost is.  The goal in doing this is that the allowed operations and cost function
should be such that they accurately reflect the cost of operations in an algorithm.  In particular, we would like that our model is powerful enough to be able to simulate all of the known
combinatorial algorithms with running time no larger than their actual running time, but not so powerful that it allows for fast (e.g. quadratic time) algorithms that are not implementable on a real computer.  We still don't have a satisfactory model with these properties.

This paper takes a step in this direction.  We develop a model which captures some of what a combinatorial algorithm might do.  In particular our model is capable of efficiently simulating the Four Russians algorithm, but is sufficiently more general.  We then prove a superquadratic  lower bound in the model: Any algorithm for BMM in this model requires time at least
$\Omega(n^{7/3} / 2^{O(\sqrt{ \log n })})$.

Unfortunately, our model is not strong enough to simulate the more recent combinatorial approaches.  Our hope is that our approach provides a starting point for a more comprehensive analysis of the limitation of combinatorial algorithms.

One of the key features of our lower bound is the identification of a family of "hard instances" for BMM.  
In particular, we use tripartite graphs on roughly $3n$ vertices that have almost quadratic number a pairs of vertices from
the first and the last layers connected by a single ({\em unique}) path via the middle layer. These graphs are derived from
{\em $(r,t)$-graphs} of Rusza and Szemeredi \cite{RS78}, which are dense bipartite graphs on $2n$ vertices that can be decomposed
into linear number of disjoint induced matchings.
More recently, Alon, Moitra Sudakov \cite{AMS12} provides strengthening of Rusza and Szemeredi's construction although they lose in the parameters that
are most relevant for us.

\subsection{Combinatorial models}

The first combinatorial model for BMM was given by Angluin \cite{Angluin76}. For the product of $\Pm \times \Q$, the model allows
to take bit-wise OR ({\em union}) of rows of the matrix $\Q$ to compute the individual rows of the resulting matrix $\Pm \Q$. The cost in this model
is the number of unions taken. By a counting argument, Angluin \cite{Angluin76} shows that there are matrices $\Pm$ and $\Q$ such that 
the number of unions taken must be $\Omega(n^2/\log n)$. This matches the number of unions taken by the Four Russians Algorithm, and in that sense the Four Russians Algorithm  is optimal.

If the cost of taking each row union were counted as $n$, the total cost would become $\Theta(n^3/\log n)$. The Four Russians Algorithm  improves this time to $O(n^3/\log^2 n)$
by leveraging ``word-level parallelism''   to compute each row union in time $O(n/\log n)$. 

A possible approach to speed-up the Four Russians Algorithm  would be to lower the cost of each union operation even further. The above analysis ignores the fact
that we might be taking the union of rows with identical content multiple times. For example if $\Pm$ and $\Q$ are  random matrices (as in the lower
bound of Angluin) then each row of the resulting product is an all-one row. Such rows will appear after taking an union of merely $O(\log n)$ rows
from $\Q$. An entirely naive algorithm would be taking unions of an all-one row with $n$ possible rows of $\Q$ after only few unions. Hence, there would
be only $O(n \log n)$ different unions to take for the total cost of $O(n^2 \cdot poly(\log n))$.
We could quickly detect repetitions of unions by maintaining a short fingerprint for each row evaluated. 

Our first model takes repetitions into account. Similarly to Angluin, we focus on the number of unions taken by the algorithm
but we charge for each union differently. The natural cost of  a union of rows with values $u,v\in\{0,1\}^n$ counts
 the cost as the minimum of the number of ones in $u$ and $v$. This is the cost we count as one could use sparse set representation for $u$ and $v$.
In addition to that if unions of the same rows (vectors) are taken multiple times we charge all of them only ones, resp. we charge the first one
the proper cost and all the additional unions are for a unit cost. As we have argued, on random matrices $\Pm$ and $\Q$, BMM will cost
$O(n^2 \log n)$ in this model. Our first lower bound shows that even in this model, there are matrices for which the cost of BMM is almost cubic.

\begin{theorem}[Informal statement]
In the row-union model with removed repetitions the cost of Boolean matrix multiplication is $\Omega(n^3 / 2^{O(\sqrt{ \log n })})$.
\end{theorem}

The next natural operation one might allow to the algorithm is to divide rows into pieces. This is indeed what the Four Russians Algorithm  and many other algorithms
do. In the Four Russians Algorithm, this corresponds to the ``word-level parallelism''. Hence we might allow the algorithm to break rows into pieces, take
unions of the pieces, and concatenate the pieces back. In our more general model we set the cost of the partition and concatenation to be a unit cost,
and we only allow to split a piece into continuous parts. More complex partitions can be simulated by performing many two-sided partitions
and paying proportionally to the complexity of the partition. The cost of a union operation is again proportional to the smaller number of ones in the 
pieces, while repeated unions are charged for a unit cost. In this model one can implement the Four Russians Algorithm  for the cost $O(n^3/\log^2 n)$, matching its usual
cost. In the model without partitions the cost of the Four Russians Algorithm  is $\Theta(n^3/\log n)$.

In this model we are able to prove super-quadratic lower bound when we restrict that all partitions happen first, then unions take place, and then
concatenations. 

\begin{theorem}[Informal statement]
In the row-union model with partitioning and removed repetitions the cost of Boolean matrix multiplication is 
$\Omega(n^{7/3} / 2^{O(\sqrt{ \log n })})$.
\end{theorem}

Perhaps, the characteristic property of ``combinatorial'' algorithms is that from the run of such an algorithm one can extract a combinatorial proof 
({\em witness}) for the resulting product. This is how we interpret our models. For given $\Pm$ and $\Q$ we construct a witness circuit that mimics
the work of the algorithm. The circuit operates on rows of $\Q$ to derive the rows of the resulting matrix $\Pm\Q$. The values flowing through the circuit
are bit-vectors representing the values of rows together with information on which union of which submatrix of $\Q$ the row represents.
The gates can partition the vectors in pieces, concatenate them and take their union. For our lower bound we require that unions take place
only after all partitions and before all concatenations. This seems to be a reasonable restriction since we do not have to emulate the run of an algorithm
step by step but rather see what it eventually produces.
Also allowing to mix partitions, unions and concatenations in arbitrary order could perhaps lead to only quadratic cost on all matrices. We are not able to argue otherwise.

The proper modelling of combinatorial
algorithms is a significant issue here: one wants a model that is strong
enough to capture known algorithms (and other conceivable algorithms) but
not so strong that it admits unrealistic quadratic algorithms. We do not
know how to do this yet, and the present paper is intended as a first
step.

\subsection{Our techniques}

Central to our lower bounds are graphs derived from $(r,t)$-graphs of Rusza and Szemeredi \cite{RS78}.
Our graphs are tripartite with vertices split into parts $A, B, C$, where $|A|=|C|=n$ and $|B|=n/3$.
The key property of these graphs is that there are almost quadratically many pairs $(a,c)\in A \times C$ 
that are connected via a single ({\em unique}) vertex from $B$. In terms of the corresponding matrices $\Pm$ and $\Q$ this
means that in order to evaluate a particular row of their product we must take a union of very specific rows in $\Q$.
The number of rows in the union must be almost linear. Since $\Q$ is dense this might lead to an almost cubic cost for the 
whole algorithm provided different vertices in $A$ are connected to different vertices in $B$ so we take different unions.

This is not apriori the case for the $(r,t)$-derived graph but we can easily achieve it by removing edges between $A$ and $B$ at random,
each independently with probability 1/2. The neighborhoods of different vertices in $A$ will be very different then.
We call such a graph {\em diverse} (see a later section for a precise definition). It turns out that for our
lower bound we need a slightly stronger property, not only that we take unions of different rows of $\Q$ but also that the results
of these unions are different. We call this stronger property {\em unhelpfulness}.

Using unhelpfulness of graphs we are able to derive the almost cubic lower bound on the simpler model.  Unhelpfulness is a much more
subtle property than diversity, and we crucially depend on the properties of our graphs to derive it.

Next we tackle the issue of lower bounds for the partition model. This turns out to be a substantially harder problem,
and most of the proof is in the appendix. One needs unhelpfulness on different pieces of rows 
(restrictions to columns of $\Q$), that is making sure that the result of union of some pieces does not appear (too often)
as a result of union of another pieces. This is impossible to achieve in full generality. Roughly speaking what we can achieve is that different
parts of {\em any} witness circuit cannot produce the same results of unions.

The key lemma that formalizes it (Lemma \ref{l-reuse}) shows that the results of unions obtained for a particular interval of columns in $\Q$
can be used at most $O(\log n)$ times on average in the rest of the circuit. This is a property of the graph which we refer to as that the graph
{\em admitting only limited reuse}. This key lemma is technically complicated and challenging to prove (albeit elementary).
Putting all the pieces together turns out to be also quite technical.

\section{Notation and preliminaries}

For any integer $k\ge 1$, $[k]=\{1,\dots,k\}$. For a vertex $a$ in a graph $G$ and a subset $S$ of vertices of $G$, 
$\Gamma(a)$ are the neighbors of $a$ in $G$, and $\Gamma_S(a)=\Gamma(a) \cap S$. (To emphasize which graph $G$ we mean we may write $\Gamma_{S,G}(a)$.)
A {\em subinterval} of $C=\{c_1,c_2,\dots,c_n\}$ is any set $K=\{c_i,c_{i+1},\dots,c_j\}$, for some $1 \le i \le j \le |C|$. 
By $\min K$ we understand $i$ and by $\max K$ we mean $j$.
For a subinterval $K=\{c_i,c_{i+1},\dots,c_{i+\ell-1}\}$ of $C$ and  a vector $v\in\{0,1\}^\ell$, $K\rest_v$ denotes the set  
$\{c_j\in K;\; v_{j-i+1}=1\}$. For a vector $v\in\{0,1\}^n$, $v\rest_K = v_i,v_{i+1},\dots,v_{i+\ell-1}$.
For a binary vector $v$, $|v|$ denotes the number of ones in $v$.

\subsection{Matrices}
We will denote matrices by calligraphic letters $\Pm,\Q,\Rm$. All matrices we consider are binary matrices.
For integers $i,j$, $\Pm_i$ is the $i$-th row of $\Pm$ and $\Pm_{i,j}$ is the $(i,j)$-th entry of $\Pm$. 
Let $\Pm$ be an $n_A \times n_B$ matrix and $\Q$ be an $n_B \times n_C$ matrix, for some integers $n_A,n_B,n_C$.
We associate matrices $\Pm,\Q$  with a tripartite graph $G$. The vertices of $G$ is the set $A \cup B \cup C$
where $A=\{a_1,\dots,a_{n_A}\}$, $B=\{b_1,\dots,b_{n_B}\}$ and $C=\{c_1,\dots,c_{n_C}\}$. The edges of $G$ 
are $(a_i,b_k)$ for each $i,k$ such that $\Pm_{i,k}=1$, and $(b_k,c_j)$ for each $k,j$ such that $\Q_{k,j}=1$.
In this paper we only consider graphs of this form.
Sometimes we may abuse notation and index matrix $\Pm$ by vertices of $A$ and $B$, and similarly $\Q$ by vertices from $B$ and $C$.
For a set of indices $S \subseteq B$, $\wrow(\Q_S) = \bigvee_{i\in S} \Q_i$ is the bit-wise $\OR$ of rows of $\Q$ given by $S$.

\subsection{Model}

\smallskip
\noindent{\bf Circuit.}
A {\em circuit} is a directed acyclic graph $W$ where each node ({\em gate})
has in-degree either zero, one or two. 
The {\em degree} of a gate is its in-degree, the {\em fan-out} is its out-degree. Degree one gates are called {\em unary}
and degree two gates are {\em binary}.
Degree zero gates are called {\em input gates}.
For each binary gate $g$, $\wleft(g)$ and $\wright(g)$ are its two predecessor gates. 
A computation of a circuit proceeds by passing values along edges, where each gate processes
its incoming values to decide on the value passed along the outgoing edges. The input gates
have some predetermined values. The output of the circuit is the output value of some designated
vertex or vertices.

\smallskip
\noindent{\bf Witness.}
Let $\Pm$ and $\Q$ be matrices of dimension $n_A \times n_B$ and $n_B \times n_C$, resp., with its associated graph $G$.
A {\em witness} for the matrix product $\Pm \times \Q$ is a circuit consisting of input gates, unary {\em partition gates},
binary union gates and binary {\em concatenation gates}. 
The values passed along the edges are triples $(S,K,v)$, where $S\subseteq B$ identifies a set of rows of the matrix $\Q$,
the subinterval $K\subseteq C$ identifies a set of columns of  $\Q$, and $v$ is the restriction $\wrow(\Q_S)\rest_K$ of $\wrow(\Q_S)$ to the columns of $K$.
Each input gate outputs $(\{b\},C,Q_b)$ for some assigned $b\in B$. A partition gate with an assigned subinterval $K'\subseteq C$
on input $(S,K,v)$ outputs {\em undefined} if $K'\not\subseteq K$ and outputs $(S,K',v')$ otherwise, where $v'\in \{0,1\}^{|K'|}$
is such that for each $j\in [|K'|]$, $v'_j = v_{j + \min K' - \min K}$. A union gate on inputs $(S_\LL,K_\LL,v_\LL)$ and 
$(S_\RR,K_\RR,v_\RR)$ from its children outputs {\em undefined} if $K_\LL\neq K_\RR$, and outputs $(S_\LL \cup S_\RR,K_\LL,v_\LL \cup v_\RR)$ otherwise.
A concatenation gate, on inputs $(S_\LL,K_\LL,v_\LL)$ and $(S_\RR,K_\RR,v_\RR)$ where $\min K_\LL \le \min K_\RR$, is {\em undefined} if
$\max K_\LL+1 < \min K_\RR$ or $S_\LL \neq S_\RR$ or $\max K_\LL > \max K_R$ and outputs $(S_\LL,K_\LL \cup K_\RR,v')$ otherwise,
where $v'$ is obtained by concatenating $v_\LL$ with the last $(\max K_\RR - \max K_\LL)$ bits of $v_\RR$.

It is straightforward that whether a gate is undefined depends solely on the structure of the circuit but not on the actual values of $\Pm$ or $\Q$. 
We will say that the circuit is {\em structured} if union gates do not send values into partition gates, and concatenation gates
do not send values into partition and union gates. Such a circuit first breaks rows of $\Q$ into parts, computes union of compatible parts
and then assembles resulting rows using concatenation.

We say that a witness $W$ is a {\em correct witness} for $\Pm \times \Q$ if $W$ is structured, no gate has undefined output, and for each $a \in A$, 
there is a gate in $W$ with output $(\Gamma_B(a),C,v)$ for $v=\wrow(\Q_{\Gamma_B(a)})$.

\smallskip
\noindent{\bf Cost.}
The {\em cost} of the witness $W$ is defined as follows. For each union gate $g$ with inputs $(S_\LL,K_\LL,v_\LL)$ and 
$(S_\RR,K_\RR,v_\RR)$ and an output $(S,K,v)$ we define its {\em row-class} to be $\class(g)=\{v,v_\LL, v_\RR\}$. 
If $T$ is a set of union gates from $W$, $\class(T)=\{ \{u,v,z\},$ $\{u,v,z\}$ is the row-class of some gate in $T\}$.
The cost of a row-class $\{u,v,z\}$ is $min \{|u|,|v|,|z|\}$.
The cost of $T$ is $\sum_{\{u,v,z\}\in \class(T)} min \{|u|,|v|,|z|\}$.
The {\em cost of witness $W$} is the number of gates in $W$ plus the cost of the set of all union gates in $W$. 

We can make the following simple observation.

\begin{proposition}\label{p-unionw}
If $W$ is a correct witness for $\Pm \times \Q$, then for each $a \in A$, there exists a collection of subintervals $K_1,\dots,K_\ell \subseteq C$
such that $C=\bigcup_i K_i$ and for each $i\in[\ell]$, there is a union gate in $W$ which outputs $(\Gamma_B(a),K_i,\wrow(\Q_{\Gamma_B(a)})\rest_{K_i})$.
\end{proposition}

\smallskip
\noindent{\bf Union and resultant circuit.}
One can look at the witness circuit from two separate angles which are captured in the next definitions.
A {\em union circuit over a universe $B$} is a circuit with gates of degree zero and two where
each gate $g$ is associated with a subset $\wset(g)$ of $B$ so that for each binary gate $g$,
$\wset(g)=\wset(\wleft(g))) \cup \wset(\wright(g))$. 
For integer $\ell\ge 1$, a {\em resultant circuit} is a circuit with gates of degree zero and two where
each gate $g$ is associated with a vector $\wrow(g)$ from $\{0,1\}^\ell$ so that for each binary gate $g$,
$\wrow(g) = \wrow(\wleft(g)) \vee \wrow(\wright(g))$, where $\vee$ is a coordinate-wise $\OR$. 

For a vertex $a \in A$ and a subinterval $K=\{c_i,c_{i+1},\dots,c_{i+\ell-1}\}$ of $C$,
a {\em union witness} for $(a,K)$ is a union circuit $W$ over $B$ with a single output gate $g_\out$ where $\wset(g_\out) = \Gamma_B(a)$
and for each input gate $g$ of $W$, $\wset(g) = \{b\}$ for some $b\in B$ connected to $a$.

\smallskip
\noindent{\bf Induced union witness.}
Let $W$ be a correct witness for $\Pm \times \Q$. 
Pick $a \in A$ and a subinterval $K \subseteq C$. Let there be a union gate $g$ in $W$ with output $(\Gamma_B(a),K,\wrow(\Q_{\Gamma_B(a)})\rest_{K})$.
An {\em induced union witness} for $(a,K)$ is a union circuit over $B$ whose underlying graph consists of copies of the union gates
that are predecessors of $g$, and a new input gate for each input or partition gate that feeds into one of the union gates. They are connected
in the same way as in $W$. For each gate $g$ in the induced witness we let $\wset(g)=S$ whenever its corresponding gate in $W$ outputs 
$(S,K',v)$ for some $K'$ and $v$. From the correctness of $W$ it follows that each such $K'=K$ and the resulting circuit is a correct
union witness for $(a,K)$.

%
%
%

\subsection{$(r,t)$-graphs}\label{sec-rtgraphs}
We will use special type of graphs for constructing matrices which are hard for our combinatorial model of Boolean matrix multiplication.
For integers $r,t\ge 1$, an $(r,t)$-graph is a graph whose edges can be partitioned into $t$ pairwise disjoint induced matchings of size $r$.
Somewhat counter-intuitively as shown by Rusza and Szemeredi \cite{RS78} there are dense graphs on $n$ vertices that are $(r,t)$-graphs
for $r$ and $t$ close to $n$.

\begin{theorem}[Rusza and Szemer\'edi \cite{RS78}]
For all large enough integers $n$, for $\delta_n = 1/2^{\Theta(\sqrt{\log n})}$ there is a $(\delta_n n, n/3)$-graph $G^{r,t}_n$.
\end{theorem}

A more recent work of Alon, Moitra Sudakov \cite{AMS12} provides a construction of a $(r,t)$-graphs on $n$ vertices with $rt=(1-o(1))\binom{n}{2}$ and $r=n^{1-o(1)}$.
The graphs of Rusza and Szemer\'edi are sufficient for us.

Let $G^{r,t}_n$ be the graph from the previous theorem and let $M_1,M_2,\dots, M_{n/3}$ be the disjoint induced matchings of size $\delta_n n$.
We define a tripartite graph $G_n$ as follows: $G_n$ has vertices $A=\{a_1,\dots,a_{n}\}$, $B=\{b_1,\dots,b_{n/3}\}$ and $C=\{c_1,\dots,c_{n}\}$.
For each $i,j,k$ such that $(i,j)\in M_k$ there are edges $(a_i,b_k)$ and $(b_k,c_j)$ in $G_n$. The following immediate lemma 
states one of the key properties of $G_n$.

\begin{lemma}
If $(i,j)\in M_k$ in $G^{r,t}_n$ then there is a unique path between $a_i$ and $c_j$ in $G_n$.
\end{lemma}

For the rest of the paper, we will fix the graphs $G_n$. Additionally, we will also use a graph $\G_n$ which 
is obtained from $G$ by removing each edge between $A$ and $B$ independently at random with probability $1/2$. (Technically, $\G_n$
is a random variable.) When $n$ is clear from the context we will drop the subscript $n$. 

Fix some large enough $n$.
Let $\Pm$ be the $n \times n/3$ adjaceny matrix between $A$ and $B$ in $G$ and $\Q$ be the $n/3 \times n$ adjacency matrix between $B$ and $C$
in $G$. The adjacency matrix between $A$ and $B$ in $\G$ will be denoted by $\Pt$. ($\Pt$ is also a random variable.) 
The adjacency matrix between $B$ and $C$ in $\G$ is  $\Q$.

We say that $c$ is {\em unique} for $A$ if there is exactly one 
$b\in B$ such that $(a,b)$ and $(b,c)$ are edges in $G$. The previous lemma implies that on average $a$ has many unique vertices $c$ in $G_n$, namely $\delta_n n /3$.
For $S\subseteq C$,  let $S[a]$ denote the set of vertices from $S$ that are unique for $a$ in $G$. E.g., $C[a]$ are all vertices unique for $a$.
Let $\beta_a(S)$ denote the set of vertices from $B$ that are connected to $a$ and some vertex in $S[a]$.
Notice, $|\beta_a(S)| = |S[a]|$.
Since $\beta_a(\cdot)$ and $\cdot[a]$ depend on edges in graph $G$, to emphasise 
which graph we have in mind we may subscript them by $G$: $\beta_{a,G}(\cdot)$ and $\cdot[a]_G$.

For the randomized graph $\G$ we will denote by $S[a]'_\G$ the set of vertices from $S$ that are unique for $a$ in $G$ and that are connected to $a$ via $B$
also in $\G$.
(Thus, vertices from $S$ that are not unique for $a$ in $G$ but became unique for $a$ in $\G$ are not included in $S[a]'_\G$.)
Let $\beta'_{a,\G}(S)$ denotes $\beta_a(S[a]'_\G)$

\subsection{Diverse and unhelpful graphs}

In this section we define two properties of $\G$ that capture the notion that one needs to compute many different unions of rows of $\Q$ to calculate $\Pt \times \Q$.
The simpler condition stipulates that neighborhoods of different vertices from $A$ are quite different. The second condition stipulates that not only 
the neighborhoods of vertices from $A$ are different but also the unions of rows from $\Q$ that correspond to these neighborhoods are different.

Let $G_n$ and $\G_n$ and $\Pm,\Q,\Pt$ be as in the previous section. 
For integers $k,\ell \ge 1$, we say $\G$ is {\em $(k,\ell)$-diverse} if for every set $S \subseteq B$ of size at least $\ell$,
no $k$ vertices in $A$ are all connected to all the vertices of $S$.

\begin{lemma}\label{l-diversity}
Let $c,d\ge 4$ be integers.
The probability that $\G_n$ is $(c \log n,d \log n)$-diverse is at least $1- n^{- (cd/2) \log n}$.
\end{lemma}

\begin{proof}
Let $k=c \log n$ and $\ell=d \log n$.
$\G$ is not $(k,\ell)$-diverse if for some set $S \subseteq B$ of size $\ell$,
and some $k$-tuple of distinct vertices $a_1,\dots,a_k\in A$, each vertex $a_i$ is connected to all vertices from $S$ in $\G$.
The probability that all vertices of a given $k$-tuple $a_1,\dots,a_k\in A$ are connected
to all vertices in $S$ in $\G$ is at most $2^{-k \ell}$. (The probability is zero if some $a_i$ is not connected to
some vertex from $S$ in $G$.) Hence, the probability that there is some set $S \subseteq B$ of size $\ell$,
and some $k$-tuple of distinct vertices $a_1,\dots,a_k\in A$ where each vertex $a_i$ is connected to all vertices from $S$ in $\G$
is bounded by:

\begin{eqnarray*}
 {n \choose c \log n} \cdot {n \choose d \log n} \cdot 2^{-cd \log^2 n} \le  n^{(c+d)\log n} \cdot 2^{-cd \log^2 n} \le  \frac{1}{n^{(cd/2) \log n}}
\end{eqnarray*}
where the second inequality follows from $c,d \ge 4$.
\end{proof}

For $S \subseteq B$, $a\in A$ and 
a subinterval $K\subseteq C$, we say that $S$ is {\em helpful for $a$ on $K$} 
if there exists a set $S' \subseteq \beta'_{a,\G}(K)$ such that $|S| \le |S'|$ and 
$C[a]_G \cap (K \rest_{\wrow(\Q_S)}) = C[a]_G \cap (K\rest_{\wrow(\Q_{S'})})$. In other words, the condition means that $\wrow(\Q_s)$ and  $\wrow(\Q_{S'})$
agree on coordinates in $K$ that correspond to vertices unique for $a$ in $G$. This is a necessary precondition for  $\wrow(\Q_S)\rest_K = \wrow(\Q_{S'})\rest_K$
which allows one to focus only on the {\em hard-core} formed by the unique vertices.
In particular, if for some $S''\subseteq \Gamma_{B,\G}(a)$ in $\G$,  $\wrow(\Q_S)\rest_K = \wrow(\Q_{S''})\rest_K$, then $S' = S'' \cap \beta'_{a,\G}(K)$
satisfies $C[a]_G \cap (K \rest_{\wrow(\Q_S)}) = C[a]_G \cap (K\rest_{\wrow(\Q_{S'})})$. (See the proof below.)

For integers $k,\ell \ge 1$, we say $\G$ is {\em $(k,\ell)$-unhelpful on $K$} if for every set $S \subseteq B$ of size at least $\ell$,
there are at most $k$ vertices in $A$ for which $S$ is helpful on $K$.

\begin{lemma}\label{l-unhelpful}
Let $c,d\ge 4$ be integers.
Let and  $K=\{c_i,c_{i+1},\dots,c_{i+\ell-1}\}$ be a subinterval of $C$. 
The probability that $\G_n$ is $(c \log n,d \log n)$-unhelpful on $K$ is at least $1- n^{- (cd/2) \log n}$.
\end{lemma}

\begin{proof}
Take any set $S\subseteq B$ of size $\ell \ge d \log n$ and arbitrary vertices $a_1,\dots,a_k \in A$ for $k=c \log n$. 
Consider $\wrow(\Q_S)\rest_K$ and some $i\in [k]$. Since edges between $B$ and $C$ are always the same in $\G$,
$\wrow(\Q_S)\rest_K$ is always the same in $\G$. 
If $S$ is helpful on $K$ for $a_i$ then there exists $S_i \subseteq \beta'_{a,\G}(K)$ such 
that $|S_i| \ge \ell$ and $C[a]_G \cap (K \rest_{\wrow(\Q_{S_i})}) = C[a]_G \cap (K\rest_{\wrow(\Q_S)})$.
It turns out that given $a_i$, the possible $S_i$ is uniquely determined by $\wrow(\Q_S)\rest_K$. 
Whenever $\wrow(\Q_S)\rest_K$ has one in a position $c$ that corresponds to a unique vertex of $a$ in $G$, 
$\wrow(\Q_{S_i})\rest_K$ must have one there as well
so the corresponding $b$ must be in $S_i$. 
Conversely, whenever $\wrow(\Q_S)\rest_K$ has zero in a position $c$ that corresponds to a unique vertex of $a$ in $G$, 
$\wrow(\Q_{S_i})\rest_K$ must have zero there as well so the corresponding $b$ is not in $S_i$. 
The probability that $S_i \subseteq \beta'_{a,\G}(K)$ is $2^{-|S_i|}$.

Hence, the probability over choice of $\G$ that $S$ is helpful for $a_i$ on $K$ is at most $2^{-\ell}$. For different $a_i$'s this probability
is independent as it only depends on edges between $a_i$ and $B$. Thus the probability that $S$ is helpful for $a_1,\dots,a_k$ is at most $2^{-\ell k}$.

There are at most ${n \choose \ell} \cdot {n \choose k}$ choices for the set $S$ of size $\ell$ and $a_1,\dots,a_k$. Hence, the probability
that $\G$ is not $(c \log n,d \log n)$-unhelpful on $K$ is at most:

\begin{eqnarray*}
\sum_{\ell=d \log n}^n {n \choose \ell} \cdot {n \choose k} \cdot 2^{-\ell k} \le \sum_{\ell=d \log n}^n {n^\ell} \cdot {n^k} \cdot 2^{-\ell k} \\
\le \sum_{\ell=d \log n}^n 2^{(\ell + k) \log n - \ell k} \\
\le \sum_{\ell=d \log n}^n 2^{- \ell k /2} \\
\le \sum_{\ell=d \log n}^n \frac{1}{n^{(cd/2) \log n}}
\end{eqnarray*}
where the third inequality follows from $c,d \ge 4$.
\end{proof}

\section{Union circuits}

Our goal is to prove the following theorem:

\begin{theorem}\label{tt-union}
There is a constant $c>0$ such that for all $n$ large enough there are matrices $\Pm \in \{0,1\}^{n\times n/3}$ and $\Q \in \{0,1\}^{n/3\times n}$
such that any correct witness for $\Pm \times \Q$ consisting of only union gates has cost at least $n^3 / 2^{c \sqrt{ \log n }}$.
\end{theorem}

Here by consisting of only union gates we mean consisting of union gates and input gates.
Our almost cubic lower bound on the cost of union witnesses is an easy corollary to the following lemma.

\begin{lemma}\label{l-simple}
Let $n$ be a large enough integer and $\G_n$ be the graph from Section \ref{sec-rtgraphs}, and $\Pt,\Q$ be its corresponding matrices. 
Let $W$ be a correct witness for $\Pt \times Q$ consisting of only union gates. 
Let $\Pt$ have at least $m$ ones. Let each row of $\Q$ have at least $r$ ones.
If $\G$ is $(k,\ell)$-unhelpful on $C$ for some integers $k,\ell \ge 1$ then any correct witness for 
$\Pt \times \Q$ consisting of only union gates has cost at least $(mr / 2k\ell) - nr/k$.
\end{lemma}

\begin{proof}
Let $W$ be a correct witness for $\Pt \times \Q$ consisting of only union gates. 
For each gate $g$ of $W$ with output $(S,C,v)$, for some $v$, define $\wset(g)=S$. 
Consider $a \in A$. Let $g_a$ be a gate of $W$ such that $\wset(g_a)=\Gamma_{B,\G}(a)$ (which equals $\beta'_{a,\G}(C)$).
Take a maximal set $D_a$ of gates from $W$, descendants of $g_a$, such that for each $g\in D_a$, $|\wset(g)| \ge \ell$ and either $|\wset(\wleft(g))| < \ell$
or  $|\wset(\wright(g))| < \ell$, and furthermore for $g\neq g' \in D_a$, 
$\{\wset(g),\wset(\wleft(g)),\wset(\wright(g))\} \neq \{\wset(g'),\wset(\wleft(g')),\wset(\wright(g'))\}$.

Notice, if $g\neq g' \in D_a$ then $\class(g) \neq \class(g')$. This is because for any sets $S \neq S' \subseteq \wset(g_a)$,
$\wrow(\Q_S) \neq \wrow(\Q_{S'})$. (Say, $b\in S \setminus S'$, then there is 1 in $\Q_b$ which corresponds to a vertex $c$ unique for $a$.
Thus, $\wrow(\Q_S)_c = 1$ whereas $\wrow(\Q_{S'})_c = 0$.)

We claim that since $D_a$ is maximal, $|D_a| \ge \lfloor |\wset(g_a)| / 2\ell \rfloor$. 
We prove the claim. Assume $\wset(g_a) \ge 2\ell$ otherwise there is nothing to prove.
Take any $b \in \wset(g_a)$ and consider a path $g_0,g_1,\dots,g_p=g_a$ of gates in $W$ such that $\wset(g_0) = \{ b \}$.
Since $|\wset(g_0)|=1$, $|\wset(g_a)|\ge 2\ell$ and $\wset(g_{i-1}) \subseteq \wset(g_i)$, there is some $g_i$ with 
$|\wset(g_i)|\ge \ell$ and $|\wset(g_{i-1})| < \ell$. By maximality of $D_a$ there is some gate $g \in D_a$ such that
$\{\wset(g),\wset(\wleft(g)),\wset(\wright(g))\} = \{\wset(g_i),\wset(\wleft(g_i)),\wset(\wright(g_i))\}$.
Hence, $b$ is in $\wset(\wleft(g))$ or $\wset(\wright(g))$ of size $<\ell$.
Thus
$$
\wset(g_a) \subseteq \bigcup_{g \in D_a;\; |\wset(\wleft(g))| < \ell} \wset(\wleft(g)) \cup
                    \bigcup_{g \in D_a;\; |\wset(\wright(g))| < \ell} \wset(\wright(g))
$$
Hence, $|\wset(g_a)| \le 2\ell \cdot |D_a|$ and the claim follows.

For a given $a$, gates in $D_a$ have different row-classes. Since $\G$ is $(k,\ell)$-unhelpful on $C$, the same row-class can appear 
in $D_a$ only for at most $k$ different $a$'s. (Say, there were $a_1,a_2,\dots,a_{k+1}$ vertices in $A$ and gates $g_1\in D_{a_1}, \dots, g_{k+1}\in D_{a_{k+1}}$ of the same row-class.
For each $i\in[k+1]$, $\wset(g_i)\subseteq \Gamma_{B,\G}(a_i) = \beta'_{a_i,\G}(C)$
and $|\wset(g_i)| \ge \ell$. The smallest $\wset(g_i)$ would be helpful for $a_1,a_2,\dots,a_{k+1}$ contradicting the unhelpfulness of $\G$.)
Since
$$
\sum_a |D_a| \ge \sum_a \lfloor |\wset(g_a)| / 2\ell \rfloor \ge \frac{m}{2\ell} - n,
$$
witness $W$ contains gates of at least $(m/2k\ell)-n/k$ different row-classes. Since, each $\Q_b$
contains at least $r$ ones, the total cost of $W$ is as claimed.
\end{proof}

\begin{proofof}{Theorem \ref{tt-union}}
Let $\G_n$ be the graph from Section \ref{sec-rtgraphs}, 
and $\Pt,\Q$ be its corresponding matrices.  Let $r=n \delta_n$. 
By Lemma~\ref{l-unhelpful}, the graph $\G$ is $(5\log n,5\log n)$-unhelpful on $C$
with probability at least $1-1/n^{\log n}$, and by Chernoff bound, $\Pt$ contains at least $nr/10$ ones
with probability at least $1-{\mathrm exp}(n)$.  So with probability at least $1/2$, $\Pt$ has $m \ge nr/10$ ones
while $\G$ is $(5\log n,5\log n)$-unhelpful on $C$. By the previous lemma, any witness for $\Pt \times \Q$
is of cost $(nr^2/25 \log n) - nr/5 \log n$. For large enough $n$, this is at least 
$nr^2 / 50 \log n = n^3 \delta_n^2/ 50 \log n$, and the theorem follows.
\end{proofof}

\section{Circuits with partitions}

In this section, our goal is to prove the lower bound $\Omega(n^{7/3} / 2^{O(\sqrt{ \log n })})$ on the cost of a witness for matrix product
when the witness is allowed to partition the columns of $\Q$. Namely:

\begin{theorem}\label{t-union}
For all $n$ large enough there are matrices $\Pm \in \{0,1\}^{n\times n/3}$ and $\Q \in \{0,1\}^{n/3\times n}$
such that any correct witness for $\Pm \times \Q$ has cost at least $\Omega(n^{7/3} / 2^{O(\sqrt{ \log n })})$.
\end{theorem}

We provide a brief overview of the proof first. The proof builds on ideas seen already in the previous part but also requires several additional ideas. 
Consider a correct witness for $\Pt \times \Q$. We partition its union gates based on their corresponding subinterval of $C$.
If there are many vertices in $A$ that use many different subintervals (roughly $\Omega(n^{4/3})$ in total) the lower bound follows by 
counting the total number of gates in the circuit using diversity of $\G$ (Lemma~\ref{l-manyparts}). If there are many vertices in $A$ which use
only few subintervals (less than roughly $O(n^{1/3})$ each) then these subintervals must be large on average (about $n^{2/3}$) and contain
lots of vertices from $C$ unique for their respective vertices from $A$. 

In this case we divide the circuit (its union gates) based on their subinterval, and we calculate the contribution of each 
part separately. To do that we have to limit the amount of reuse of a given row-class within each part, and also among 
distinct parts. Within each part we limit the amount of reuse using a similar technique to Lemma~\ref{l-simple} based on
unhelpfulness of the graph (Lemma~\ref{l-unhelpfulcost}). 
However, for distinct parts we need a different tool which we call {\em limited reuse}. 
Limited reuse is somewhat different than unhelpfulness in the type of guarantee we get. It is a weaker guarantee
as we are not able to limit the reuse of a row-class for each single gate but only the total reuse of row-classes
of all the gates in a particular part. On average the reuse is again roughly $O(\log n)$.

However, the number of gates in a particular part of the circuit might be considerably larger than the number of 
gates we are able to charge for work in that part. In general, we are only able to charge gates that already
made some non-trivial progress in the computation (as otherwise the gates could be reused heavily.) 
We overcome this obstacle by balancing the size of the part against the number of {\em chargeable} gates in that part.

If the total number of gates in the part is at least $n^{1/3}$-times larger than the total number of chargeable gates, we charge the part 
for its size. Otherwise we charge it for work. Each chargeable gates contributes by about $n^{2/3}$ units of work or more,
however this can be reused almost $n^{1/3}$-times elsewhere. Either way, approximately $\Omega(n^{7/3})$ of work must be done in total.
Now we present the actual proof.

In order to prove the theorem we need few more definitions.
Let $G_n$ and $\G_n$ and $\Pm,\Q,\Pt$ be as in the Section \ref{sec-rtgraphs}.
All witness circuits in this section are with respect to $\Pt \times \Q$ (i.e., $\G_n$). Let $c_0$ and $c_1$ be some constants that we will fix later.

The following definition aims to separate contribution from different rows within a particular subcircuit.
A witness circuit may benefit from taking a union of the same row of $\Q$ multiple times to obtain a particular union.
This could help various gates to attain the same row-class. In order to analyze the cost of the witness we want
to effectively prune the circuit so that contribution from each row of $\Q$ is counted at most once. The following definition
captures this prunning.

Let $W$ be a union circuit over $B$ with a single vertex $g_\out$ of out-degree zero ({\em output gate}). The {\em trimming} of $W$ is
a map that associates to each gate $g$ of $W$ a subset $\wtrim(g)\subseteq \wset(g)$ such that $\wtrim(g_\out)=\wset(g_\out)$ and 
for each non-input gate $g$, $\wtrim(g) = \wtrim(\wleft(g)) \dot\cup \wtrim(\wright(g))$. For each circuit $W$, 
we fix a canonical trimming that is obtained from $\wset(\cdot)$ by the following process:
For each $b\in \wset(g_\out)$, find the left-most path from $g_\out$ to an input gate $g$ such that $b\in\wset(g)$, and
remove $b$ from $\wset(g')$ of every gate $g'$ that is {\em not} on this path.  

Given the trimming of a union circuit $W$ we will focus our attention only on gates that contribute substantially to the cost of the computation. We
call such gates {\em chargeable} in the next definition.
For a vertex $a \in A$ and a subinterval $K\subseteq C$, let $W$ be a union witness for $(a,K)$ with its trimming.
We say a gate $g$ in $W$ is {\em $(a,K)$-chargeable} if $|\wtrim(g) \cap \beta'_{a,\G}(K) | \ge c_0 \log n$
and $\wtrim(\wleft(g)) \cap \beta'_{a,\G}(K)$ and $\wtrim(\wright(g)) \cap \beta'_{a,\G}(K)$ are both different from $\wtrim(g) \cap \beta'_{a,\G}(K)$. {\em $(a,K)$-Chargeable descendants of $g$} 
are $(a,K)$-chargeable gates $g'$ in $W$ where $\wtrim(g') \cap \beta'_{a,\G}(K) \subseteq \wtrim(g) \cap \beta'_{a,\G}(K)$. 
Observe that the number of $(a,K)$-chargeable descendants of a gate $g$ is at most $|\wtrim(g) \cap \beta'_{a,\G}(K) | + 1 - c_0 \log n$.

From a correct witness for $\Pt \times \Q$, we extract some induced union circuit $W$ for $(a,K)$ and some resultant circuit $W'$.
We say that a gate $g$ from $W$ {\em is compatible} with a gate $g'$ from $W'$ if $\wrow(\Q_{\wset(g)})\rest_K = \wrow(g')$.

We want to argue that chargeable gates corresponding to gates of a given correct witness have many different row-classes. Hence, we want to bound the number of gates
whose result is compatible with each other. This is akin to the notion of helpfulness.
In the case of helpfulness we were able to limit the repetition of the same row-class for individual gates operating on the same subinterval
of columns of $\Q$. In addition to that we need to limit the occurence of the same row-class for gates that operate on distinct subintervals.
As opposed to the simpler case of helpfulness, we will need to focus on the global count of row-classes that can be reused elsewhere
from gates operating on the same subinterval. 
The next definition encapsulates the desired property of $\G$.

For $a,a'\in A$ and subintervals $K,K'$ of $C$, we say that $(a,K)$ and $(a',K')$ are {\em independent} if either $a\neq a'$ or $K \cap K' = \emptyset$.
A resultant circuit $W'$ over $\{0,1\}^{\ell}$ is consistent with $\Q$, if there exists a subinterval $K\subseteq C$ of size $\ell$, 
such that for each input gate $g$ of $W'$, $\wrow(g)=\Q_b\rest_K$ for some $b\in B$. 
We say that $\G$ {\em admits only limited reuse} if for any resultant circuit $W'$ of size at most $n^3$ which is consistent with $\Q$ 
and any correct witness circuit $W$ for $\Pt \times \Q$, the number of gates in any induced union witnesses $W_1,\dots, W_s$ for any pairwise independent
pairs $(a_1,K_1),\dots,(a_s,K_s)$ that are chargeable and compatible with some gate in $W'$ is at most $c_1 |W'| \log n$.

We will show that with high probability $\G$ admits only limited reuse.

\begin{lemma}\label{l-reuse}
Let $c_1 \ge 7$ and $c_0 \ge 20$ be constants. 
Let $n$ be a large enough integer.
Let $\G_n$ be the graph from Section \ref{sec-rtgraphs}, and $\Pt,\Q$ be its corresponding matrices. 
The probability that $\G$ admits only limited reuse is at least $1-1/n$.
\end{lemma}

To prove this lemma we will analyze individual pairs $(a,K)$ and their induced union circuits. 

\begin{lemma}\label{l-randomization}
Let $c_0 \ge 5$ be a constant. 
Let $1 \le m, \ell \le n$ be integers. 
Let $W'$ be arbitrary resultant circuit over $\{0,1\}^\ell$ with at most $n^3$ gates.
Let $a\in A$ and  $K$ be a subinterval of $C$ of size $\ell$. 
Let $E_m$ be the event that there is a union witness $W$ for $(a,K)$ in which at least $m$ $(a,K)$-chargeable
gates are compatible with gates in $W'$. There exists another event $E'_m$ that depends only on the presence or absence
of edges between $a$ and $\beta'_{a,\G}(K)$ in $\G$ such that $E_m$ implies $E'_m$, and the probability $\Pr[E_m]\le \Pr[E'_m] \le 2^{-m-(c_0-5) \log n}$. 
\end{lemma}

For independent pairs $(a,K)$ and $(a',K')$, the events $E'_m$ from Lemma~\ref{l-randomization} are independent so we will be able to bound the probability of them occuring simultaneously.

\begin{proof}
We claim that if $E_m$ occurs then there must be a $t$-tuple of gates $g'_1,\dots,g'_t$ in $W'$, where $t \le n^3$, such that the set:
$$X = \bigcup_{j=1}^t \beta'_{a,\G}(K \rest_{\wrow(g'_j)})$$ satisfies
\begin{enumerate}
\item $|X| \ge m + t (c_0-1) \log n$, and
\item edges between $a$ and $X$ are all present in $\G$.
\end{enumerate}
The existence of such a $t$-triple is our event $E'_m$. $E'_m$ has probability at most 
$$ \sum_{t=1}^{n^3} |W'|^t \cdot 2^{- (m + t (c_0-1) \log n) } \le 2^{-m} \cdot \sum_{t=1}^{n^3} 2^{- t (c_0-4) \log n) } \le 2 \cdot 2^{-m-(c_0-4) \log n},$$
as there are $|W'|^t$ choices for the $t$-tuple $g'_1,\dots,g'_t$, and the probability that all edges between $a$ and $X$ are present in  $\G$
is $2^{-|X|}$. The lemma follows in such a case as $E'_m$ only depends on the presence or absence of edges between $a$ and $\beta'_{a,\G}(K)$ in $\G$. 
So we only need to prove the existence of the $t$-tuple of required properties whenever $E_m$ occurs.

Let $S$ be a set of $m$ $(a,K)$-chargeable gates in $W$ which are compatible with some gate in $W'$. 
For each gate $g\in S$, let $\wtrim'(g) = \wtrim(g) \cap \beta'_{a,\G}(K)$. Let $g_1,\dots,g_s$
be all the gates in $S$ that are maximal with respect to inclusion of their sets $\wtrim'(g_i)$. All the gates in $S$ are among the chargable
descendants of $g_1,\dots,g_s$. Observe:
\begin{enumerate}
\item For any $i\neq j \in [s]$, $\wtrim'(g_i) \cap \wtrim'(g_j) = \emptyset$, and
\item for any  $i \in [s]$, the number of $(a,K)$-chargeable descendants of $g_i$ is at most $|\wtrim'(g_i)| + 1 - c_0 \log n$.
\end{enumerate}
The first item holds as $\wtrim(g_i)$ are either related by inclusion or disjoint, the second item holds by the definition of $(a,K)$-chargeable gates.
This implies:
$$|S| = m \le \left( \sum_{i=1}^s | \wtrim'(g_i) | \right) + s - s c_0 \log n.$$

Pick the smallest set of gates $g'_1,\dots,g'_t$ in $W'$ so that each of the gates $g_1,\dots,g_s$ is compatible with at least one of them. Clearly, $t \le s$.
Let $g_\out$ be the top-most gate of $W$. By definition, $\wset(g_\out) = \Gamma_{B,\G}(a)$.
If $g_i$ is compatible with $g'_j$ then $\wrow(\Q_{\wset(g_i)})\rest_K = \wrow(g'_j)$. Hence, $\beta'_{a,\G}(K \rest_{\wrow(g'_j)}) \subseteq \wset(g_i) \subseteq \wset(g_\out)$, and
$\wtrim'(g_i) \subseteq \beta'_{a,\G}(K \rest_{\wrow(g'_j)})$ by properties of vertices unique for $a$. For the set $X=\bigcup_{j=1}^t \beta'_{a,\G}(K \rest_{\wrow(g'_j)})$,
the former implies that all edges between $a$ and $X$ must be present in $\G$. The latter implies $|X| \ge \sum_{i=1}^s | \wtrim'(g_i) | \ge m + s (c_0-1) \log n
\ge m+t (c_0-1) \log n$. Hence, $g'_1,\dots,g'_t$ is a tuple of required properties and the lemma follows.
\end{proof}

\begin{proofof}{Lemma \ref{l-reuse}}
Fix arbitrary resultant circuit $W'$ of size at most $n^3$ consistent with $\Q$.
Fix $s \in [n^3]$ and pairwise independent $(a_1,K_1),(a_2,K_2),\dots,(a_s,K_s)$, where each $a_i \in A$ and $K_i$ is a subinterval of $C$.
Fix a sequence of positive integers $m_1,m_2,\dots,m_s$ such that $\sum_{i\in[s]} m_i \ge c_1 |W_1| \log n$. 

Take $\G$ at random. Let $W$ be some correct witness for $\Pt \times \Q$ which for each $i\in[s]$, contains an induced union witness $W_i$ for $(a_i,K_i)$
such that $W_i$ contains at least $m_i$ $(a_i,K_i)$-chargeable gates compatible with gates in $W'$. $W$ might not exist.
Our goal is to estimate the probability that such a union witness $W$ exists.

Let $E_i$ be the event that there is some union witness $W_i$ for $(a_i,K_i)$ which contains at least $m_i$ $(a_i,K_i)$-chargeable gates compatible with gates in $W'$.
We can associate to $E_i$ also an event $E'_i$ from Lemma~\ref{l-randomization}. Since  $(a_1,K_1),(a_2,K_2),\dots,(a_s,K_s)$ are pairwise independent,
the events $E'_i$ are mutually independent. Thus we can bound the probability of the existence of $W$ by

\begin{eqnarray*}
\Pr[W_1,W_2,\dots,W_s \;\mathrm{ exists}] &=& \Pr[E_1 \cap E_2 \dots \cap E_s]\\ &\le& 
\Pr[E'_1 \cap E'_2 \dots \cap E'_s] \\ &=& \prod_{i\in [s]} \Pr[E'_i] \\ &\le& 
\prod_{i\in [s]} 2^{-m_i-(c_0-5) \log n} \\ &\le& 2^{-c_1 |W'| \log n - s(c_0-5) \log n}
\end{eqnarray*}
where the second equality follows from the independence and the second inequality follows from Lemma~\ref{l-randomization}.

This probability is for a fixed choice of $W$, $s$, $a_i$'s, $K_i$'s, and $m_i$'s.
For a given size $t=|W'|$ there are at most $(t^2 + n)^t n^2$ choices for $W'$ consistent with $\Q$. 
There are also at most $(n^3)^s$ choices for $(a_1,K_1),\dots,(a_s,K_s)$
and at most $(c_1 n^3 \log n)^s$ choices for $m_1,\dots,m_s$.

Thus the probability that $\G$ does not admit only limited reuse is at most:
\begin{eqnarray*}
\sum_{t=1}^{n^3} \sum_{s=1}^{n^3} n^{3s+2} (t^2 + n)^t \cdot (c_1 n^3 \log n)^s \cdot  2^{-c_1 |W'| \log n - s(c_0-5) \log n} \le 1/n.
\end{eqnarray*}

\end{proofof}

\subsection{The cost of chargeable gates in a partition}

For $\Pt,\Q$ from Section \ref{sec-rtgraphs}, let $W$ be a correct witness for $\Pt \times Q$.
We say that a gate $g$ in $W$ is $(a,K)$-chargeable if $g$ corresponds to an $(a,K)$-chargeable
gate in the lexicographically first induced  union witness for $(a,K)$ in $W$.

The next lemma lower bounds the contribution of chargeable gates to the total cost of the witness.
It is similar in spirit to Lemma \ref{l-simple} and its proof is similar. It focuses on union gates dealing with a particular
subinterval $K\subseteq C$.

\begin{lemma}[Partition version]\label{l-unhelpfulcost}
Let $n$ be a large enough integer and $\G_n$ be the graph from Section \ref{sec-rtgraphs}, and $\Pt,\Q$ be its corresponding matrices. 
Let $r,k >1$ be integers and $\ell = c_0 \log n$.
Let $W$ be a correct witness for $\Pt \times Q$. 
Let $K\subseteq C$ be a subinterval.  
Let $R \subseteq B$ be such that for each $b$ in $R$, $\Q_b\rest_K$ has at least $r$ ones.
Let $A' \subseteq A$ be such that for each $a\in A'$, $| R \cap \beta'_{a,\G}(K)| \ge 2\ell$.
Let $m=\sum_{a\in A'} | R \cap \beta'_{a,\G}(K)|$.
If $\G$ is $(k,\ell)$-unhelpful on $K$ then there is a set $D$ of union gates in $W$ such that
\begin{enumerate}
\item Each gate in $D$ is $(a,K)$-chargeable for some vertex $a\in A$, and
\item The number of different row-classes of gates in $D$ of cost $\ge r$ is at least $m / 4k\ell$.
\end{enumerate}
\end{lemma}

\begin{proof}
Pick $a \in A'$ for which there is an induced union witness in $W$. Fix the lexicographically first union witness $W_a$ for $(a,K)$.
Let $\wtrim(\cdot)$ be its trimming. Define $\wtrim'(g)=\wtrim(g)\cap R \cap \beta'_{a,\G}(K)$.
For the output gate $g_a$ of $W_a$, $\wtrim'(g_a) = R \cap \beta'_{a,\G}(K)$ as $\beta'_{a,\G}(K) \subseteq \Gamma_{B,\G}(a) = \wtrim(g_a)$.
Take a maximal set $D_a$ of gates from $W_a$ such that for each $g\in D_a$, $|\wtrim'(g)| \ge \ell$, 
$\wtrim'(\wleft(g)), \wtrim'(\wright(g)) \subsetneq \wtrim'(g)$  and either $|\wtrim'(\wleft(g))| < \ell$
or  $|\wtrim'(\wright(g))| < \ell$, and furthermore for $g\neq g' \in D_a$, 
$\{\wtrim'(g),\wtrim'(\wleft(g)),\wtrim'(\wright(g))\} \neq \{\wtrim'(g'),\wtrim'(\wleft(g')),\wtrim'(\wright(g'))\}$.
Clearly, gates in $D_a$ are $(a,K)$-chargeable.

Notice, if $g\neq g' \in D_a$ then $\class(g) \neq \class(g')$. (Here we identify $g$ with its corresponding gate in $W$.) 
This is because for any sets $S \neq S' \subseteq \wtrim'(g_a)$,
$\wrow(\Q_S)\rest_K \neq \wrow(\Q_{S'})\rest_K$. 
(Say, $b\in S \setminus S'$, then there is 1 in $\Q_b\rest_K$ which corresponds to a vertex $c$ unique for $a$.
Thus, $\wrow(\Q_S)_c = 1$ whereas $\wrow(\Q_{S'})_c = 0$.) Also, if $g \in D_a$ and $\{u,v,z\}$ is its row-class then $|u|,|v|,|z| \ge r$,
since $\wtrim'(g), \wtrim'(\wleft(g)), \wtrim'(\wright(g))$ have non-empty intersection with $R$.

We claim that since $D_a$ is maximal, $|D_a| \ge \lfloor |\wtrim'(g_a)| / 2\ell \rfloor$. 
We prove the claim. Assume $\wtrim'(g_a) \ge 2\ell$ otherwise there is nothing to prove.
Take any $b \in \wtrim'(g_a)$ and consider a path $g_0,g_1,\dots,g_p=g_a$ of gates in $W_a$ such that $\wtrim'(g_0) = \{ b \}$.
Since $|\wtrim'(g_0)|=1$, $|\wtrim'(g_a)|\ge 2\ell$ and $\wtrim'(g_{i-1}) \subseteq \wtrim'(g_i)$, there is some $g_i$ with 
$|\wtrim'(g_i)|\ge \ell$ and $|\wtrim'(g_{i-1})| < \ell$. 
Say $g_{i-1}=\wleft(g_{i})$. 
Since $b \in \wtrim'(\wleft(g_{i})) \subsetneq  \wtrim'(g_i)$ and $\wtrim'(\wleft(g_i)) \dot\cup \wtrim'(\wright(g_i)) = \wtrim'(g_i)$, 
$\wtrim'(\wright(g_i)) \neq  \wtrim'(g_i)$.
By maximality of $D_a$ there is some gate $g \in D_a$ such that
$\{\wtrim'(g),\wtrim'(\wleft(g)),\wtrim'(\wright(g))\} = \{\wtrim'(g_i),\wtrim'(\wleft(g_i)),\wtrim'(\wright(g_i))\}$.
Hence, $b$ is in $\wtrim'(\wleft(g))$ or $\wtrim'(\wright(g))$ of size $<\ell$.
Thus
$$
\wtrim'(g_a) \subseteq \bigcup_{g \in D_a;\; |\wtrim'(\wleft(g))| < \ell} \wtrim'(\wleft(g)) \cup
                    \bigcup_{g \in D_a;\; |\wtrim'(\wright(g))| < \ell} \wtrim'(\wright(g))
$$
Hence, $|\wtrim'(g_a)| \le 2\ell \cdot |D_a|$ and the claim follows.

Set $D=\bigcup_a D_a$.
For a given $a$, gates in $D_a$ have different row-classes. Each of the row-classes is of cost at least $r$. Indeed, for each $g\in D_a$, 
$\wtrim'(g),\wtrim'(\wleft(g))$ and $\wtrim'(\wright(g))$ are all non-empty, so each of the $\wset(g),\wset(\wleft(g))$ and $\wset(\wright(g))$ contains
some $b$ with $|\Q_b\rest_K| \ge r$.

Since $\G$ is $(k,\ell)$-unhelpful on $K$, the same row-class can appear 
in $D_a$ only for at most $k$ different $a$'s. 
(Say, there were $a_1,a_2,\dots,a_{k+1}$ in $A$ and gates $g_1\in D_{a_1}, \dots, g_{k+1}\in D_{a_{k+1}}$ of the same row-class.
For each $i\in[k+1]$, $\wtrim'(g_i) \subseteq \wset(g_i) \cap \beta'_{a_i,\G}(K)$ so
$|\wset(g_i) \cap \beta'_{a_i,\G}(K)| \ge \ell$. The smallest $\wset(g_i) \cap \beta'_{a_i,\G}(K)$ would be helpful for $a_1,a_2,\dots,a_{k+1}$ contradicting the unhelpfulness of $\G$.)

Since
$$
\sum_{a\in A'} |D_a| \ge \sum_{a\in A'} \lfloor |\wtrim'(g_a)| / 2\ell \rfloor \ge \frac{m}{4\ell},
$$
$D$ contains chargeable gates of at least $m/4k\ell$ different row-classes with cost $\ge r$.
\end{proof}

\subsection{Large number of partitions}

If the witness for $\Pt \times Q$ involves many subintervals for many vertices we will apply the next lemma.

Let $n$ be a large enough integer and $\G_n$ be the graph from Section \ref{sec-rtgraphs} with associated matrices $\Pt,\Q$.
Let $W$ be a witness for $\Pt \times \Q$. By Proposition \ref{p-unionw} each $a \in A$ is associated with 
distinct subintervals $K_{a,1},\dots,K_{a,\ell_a} \subseteq C$, for some $\ell_a$, such that $C=\bigcup_{j\in [\ell_a]} K_{a,j}$ 
and there are union gates $g_{a,1},\dots,g_{a,\ell_a}$  in $W$ such that $g_{a,j}$ 
outputs $(\Gamma_{B,\G}(a),K_{a,j},v_{a,j})$ for some $v_{a,j}\in\{0,1\}^{|K_{a,j}|}$.

\begin{lemma}\label{l-manyparts}
Let $W$, $\ell_a$'s, $K_{a,j}$'s, $g_{a,j}$'s be as above. 
Let $c,d\ge 4$ and
$\ell,r \ge 1$ be integers where $r$ is large enough. Let $L=\{a\in A, \ell_a \ge \ell\;\&\;|\Gamma_{B,\G}(a)| \ge r\}$.
If $\G$ is $(c \log n, d \log n)$-diverse then the size of $W$ is at least $r\ell \cdot |L| /(2cd \log^2 n)$. 
\end{lemma}

\begin{proof}
If two union gates $g,g'$ have outputs $(S,K,v)$ and $(S',K',v')$, resp., where $K \neq K'$, then $g$ and $g'$ cannot have
a descendant union gate in common. (This follows from consistency of union gates.)
Consider a union gate $g$ in $W$ that outputs $(S,K,v)$, where $|S| \ge d \log n$. Let $T=\{(a,j) \in L \times [\ell], g_{a,j}$ has $g$ among its descendants$\}$. Clearly, for all $(a,j)\in T$, $g_{a,j}$ outputs  $(\Gamma_{B,\G}(a),K,v_{a,j})$ for some $v_{a,j}\in\{0,1\}^{|K|}$.
Hence, $(a,j),(a,j')\in T$ implies $j=j'$. For each $(a,j)\in T$, $S \subseteq \Gamma_{B,\G}(a)$. By $(c \log n, d \log n)$-diversity of $\G$,
$|T| \le c \log n$.

For each $(a,j)\in L \times [\ell]$, $g_{a,j}$ has at least $\lfloor |\Gamma_{B,\G}(a)|/d \log n \rfloor \ge \lfloor r/d \log n\rfloor \ge r/2d \log n$ 
distinct descendant union gates $g'$ with output $(S',K_{a,j},v')$, where $|S'| \ge d \log n$ and $v'$ is arbitrary.
(Each such $g'$ has distinct $S'$.) Each such $g'$ can be descendant of at most $c \log n$ gates $g_{a,j}$ by the bound on $T$.
Hence, there are at least $|L| \cdot \ell r/ (2cd \log^2 n)$ distinct union gates in $W$.
\end{proof}

\subsection{Density lemma}

We state here an auxiliary {\em density} lemma. The proof is standard but we include it for completeness.

\begin{lemma}\label{l-density}
Let $n,r\geq 1$ be integers. Let $K_1,\dots,K_r$ be a collection of (not necessarily distinct) subintervals 
of $[n]$.
Let $u_1\in K_1, u_2\in K_2, \dots,u_r\in K_r$ be distinct elements. Denote $U=\{u_1,\dots,u_r\}$. There are at least $r/2$
sets $K_i$ such that $|K_i \cap U| \ge |K_i|r/4n$.
\end{lemma}

\begin{proof}
Any subinterval $I$ of $[n]$ is called sparse if $|I\cap U|<|I|r/4n$. Let $I_1$, \dots, $I_k$ be the set of all sparse subintervals of $[n]$. We want to prove $|\cup_{i\in[k]}(I_i\cap U)|<r/2$. Denote $S= \cup_{i\in[k]}I_i$. Suppose $I'=\{{I'}_1$, \dots, ${I'}_{\ell}\}$ be the minimal set of sparse subintervals covering all sparse subintervals. Thus $\cup_{i\in[k]}I_i=\cup_{j\in[\ell]}{I'}_j$. We claim, any $u\in S$ is covered by at most two subintervals of $I'$. As otherwise assume there are more than two subintervals in $I'$ which contain $u$. All these intervals must have some nontrivial intersection including $u$. Among them consider the two, having the left most starting point and right most end point in $[n]$. It can be easily seen that the union of these two intervals covers all the other intervals and hence the minimality of $I'$ is violated. Therefore our claim follows. Now as $|S|\leq n$, from the previous claim we get $\sum_{j\in[\ell]}|{I'}_j|\leq 2n$. The construction also implies $\bigcup_{i\in[k]}(I_i\cap U)= \bigcup_{j\in[\ell]}(I'_j\cap U)$. By the sparsity of the intervals, there are at most $2n\times r/4n=r/2$ elements of $U$ contained in $\bigcup_{j\in[\ell]}(I'_j\cap U)$ and therefore in $\bigcup_{i\in[k]}(I_i\cap U)$. 
 Thus each of the $r/2$ elements of $U\setminus S$ are contained only in subintervals which are not sparse. Hence each set $K_i$ associated with these $r/2$ elements satisfies $|K_i \cap U|\ge |K_i|r/4n$.     
\end{proof}

\subsection{The main proof}

In this section we prove the lower bound $\approx n^{7/3}$ on the cost of witnesses for matrix product.

\begin{theorem}\label{t-union}
For all $n$ large enough there are matrices $\Pm \in \{0,1\}^{n\times n/3}$ and $\Q \in \{0,1\}^{n/3\times n}$
such that any correct witness for $\Pm \times \Q$ has cost at least $\Omega(n^{7/3} / 2^{O(\sqrt{ \log n })})$.
\end{theorem}

Let $n$ be large enough and let $\G_n$ be the graph from Section \ref{sec-rtgraphs}. 
Set $c=5, d=5, c_0=7 , c_1=20 $. Let $r = \delta_n n$, $s=n^{1/3} $, $ \ell=n^{1/3}$.
With probability at least 1/2, $\G$ is simultaneously $(c \log n,d \log n)$-diverse (Lemma~\ref{l-diversity}), 
$(c \log n,d \log n)$-unhelpful on each of the ${n \choose 2}$ subintervals of $C$ (Lemma \ref{l-unhelpful}),
admits only limited reuse (Lemma \ref{l-reuse}), 
and $\sum_{a\in A} |C[a]'_\G| \ge nr/3$ (by Chernoff inequality).

Let $W$ be a correct witness for $\G$, our goal is to lower bound its cost.

We will define a sequence of sets $T_6 \subseteq T_5 \subseteq \cdots \subseteq T_1 \subseteq A \times C$ of pairs of $(a,c)$ where $c$ is unique
for $a$.

\begin{enumerate}
\item {\bf (Unique pairs.)} $T_1 = \{ (a,c), a\in A, c\in C[a]'_\G\}$ is the set of pairs of $a$ and its unique vertices. By assumption, $|T_1| \ge nr/3$.

\item {\bf (Removing sparse $a$'s.)} Let $A_2 =  \{a\in A,\, |\Gamma_{B,\G}(a)| \ge r/6\}$. Clearly, $|A_2| \ge r/6$. 
Let $T_2 = T_1 \cap (A_2 \times C) = \{(a,c)\in T_1,\, |\Gamma_{B,\G}(a)| \ge r/6\}$. 
By an averaging argument, $|T_2| \ge nr/6$.

\item {\bf (Removing $a$'s with many subintervals $K$.)} 
For each $a\in A_2$, let $K_{a,1},\dots,K_{a,\ell_a}$ be obtained from Proposition~\ref{p-unionw}.
Let $A_3=\{a\in A_2, \ell_a \le \ell\}$ and $A'_2 = A_2 \setminus A_3$. If $|A'_2| \ge r/12$ we apply 
Lemma~\ref{l-manyparts} to conclude that the size of $W$ is at least 
$\frac{r}{12}\cdot \frac{r}{6} \cdot \frac{\ell}{2cd \log^2 n} \ge r^2\ell / 150 \log^2 n$. In this case we are done.

Otherwise consider the case $|A_3| \ge r/12$. Let $T_3 = T_2 \cap (A_3 \times C)$. Since $|A'_2| < r/12$, $|T_3| \ge nr/12$.

\item {\bf (Removing small subintervals $K$.)}  For each $a \in A_3$, let $K'_{a,1},\dots,K'_{a,\ell'_a}$ be 
the subsequence of $K_{a,1},\dots,K_{a,\ell_a}$ obtained by removing each $K_{a,j}$ of size smaller than $r / 24 \ell$.
So $\ell'_a \le \ell_a \le \ell$, and each $|K'_{a,j}| \ge r/24 \ell$.

We remove pairs $(a,c)$ from $T_1$ not covered by large $K_{a,j}$'s:
Let $T_4 = T_3 \cap (\bigcup_{a \in A_3} \{a\} \times (\bigcup_{j\in [\ell'_a]} K'_{a,j} ))$.
By the size and number of the removed subintervals $K$, $|T_4| \ge nr/24$.

\item {\bf (Removing overlapping subintervals $K$.)} For each $a\in A_3$ find a collection of disjoint subintervals $K''_{a,1},\dots,K''_{a,\ell''_a}$
such that $|T_4 \cap (\{a\} \times \bigcup_{j\in [\ell''_a]} K''_{a,j} )| \ge |T_4 \cap (\{a\} \times \bigcup_{j\in [\ell'_a]} K'_{a,j} )|/2$.
(Such a collection exists: Take the smallest subcollection of $K'_{a,1},\dots,K'_{a,\ell'_a}$ which covers their entire union.
Each point from $T_4 \cap (\{a\} \times \bigcup_{j\in [\ell'_a]} K'_{a,j} )$ is contained in at most two intervals of this subcollection.
Order the subcollection by the smallest element in each interval. 
Either the subset of intervals on odd positions in this ordering or on even positions has the required property.)

Let $T_5 = T_4 \cap (\bigcup_{a \in A_3} (\{a\} \times \bigcup_{j\in [\ell''_a]} K''_{a,j} ))$.
By the choice of removed subintervals $K$, $|T_5| \ge nr/48$.

\item {\bf (Disregarding sparse sub-rows of $\Q$.)} For $b\in B$, let $T_b = \{(a,c)\in T_5,\, \{b\} = \beta_a(\{c\})$, 
i.e. $b$ is on the path between $a$ and $c\}$. Let $K(a,c)$ denote $K''_{a,j}$ such that $c\in K''_{a,j}$. (This is uniquely defined
as $K''_{a,j}$'s are disjoint.)

Set $B_6=\{b\in B,\, |T_b| \ge r/48\}$. For $b\in B_6$, $(a,c)\in T_b$, we say that the triple $(b,a,c)$, is {\em dense} 
if $|\Q_b\rest_{K(a,c)}| \ge \frac{r}{24\ell} \cdot \frac{r}{48} \cdot \frac{1}{4n}$. By Lemma~\ref{l-density},
for at least half of the pairs $(a,c)\in T_b$, $(b,a,c)$ is dense.

Let $T_6 = \bigcup_{b \in B_6} \{(a,c)\in T_5,\, (b,a,c) \mbox{ is dense}\}$.

There are at most $\frac{r}{48} \cdot \frac{n}{3}$ pairs removed from $T_5$ because $b\not\in B_6$ and at most half of the remaining points
afterwards. So $|T_6| \ge |T_5|/3 \ge nr/150$.
\end{enumerate}

Given sets $T_6$ and $A_3, B_6$ obtained so far we proceed with the final calculation.

Consider a subinterval $K \subseteq C$. Let $A_K = \{a \in A_3,\, |T_6 \cap (\{a\}\times K)| \ge 2c_0 \log n\}$, 
and $R_K = \{b \in B,\, |\Q_b\rest_K| \ge \frac{r^2}{4800\cdot n\ell}\}$.

Let $m_K = \sum_{a \in A_K} |R_K \cap \beta'_{a,\G}(K)|$ and $m'_K = \sum_{a \in A_K} |T_6 \cap (\{a\} \times K)|$.
Since for any $a \in A_3$, $|T_6 \cap (\{a\} \times K)| \le |R_K \cap \beta'_{a,\G}(K)|$, $m'_k \le m_k$.
Also, $\sum_K m_K \ge \sum_K m'_K \ge |T_6| - 2c_0 n\ell \log n \ge |T_6|/2$.

Let $s_K$ be the number of union gates in $W$ that correspond to $K$ (i.e., that output $(S,K,v)$ for some $S$ and $v$.)

Consider subintervals $K \subseteq C$, where $s m_K \le s_K$, $\C = \{K \subseteq C,\, K \mbox{ subinterval}, s m_K \le s_K\}$.
If $\sum_{K \in \C} m_K \ge |T_6|/4$ then the $|W| \ge snr/600$ so we are done. 

So consider the case when $\sum_{K \in \C'} m_K \ge |T_6|/4$, where $\C' = \{K \subseteq C,\, K \mbox{ subinterval}, s m_K > s_K\}$.
For each $K\in \C'$, apply Lemma~\ref{l-unhelpfulcost} 
($R \leftarrow R_K, A' \leftarrow A_K, k \leftarrow c\log n, \ell \leftarrow c_0\log n, r \leftarrow r^2/4800 n\ell, D \rightarrow D_K$) 
to obtain the set $D_K$ of gates with at least $m_K/4 cc_0\log^2 n$ row-classes of cost at least $r^2/4800 n\ell$.
As all the gates in $D_K$ are $(a,K)$-chargeable for some $a \in A$, by definition of limited reuse, their row-class coincides with at most 
$c_1 s_K \log n \le m_K s c_1 \log n$ other gates in $\bigcup_{K'\in \C'} D_{K'}$. 
Thus, $\bigcup_{K'\in \C'} D_{K'}$ contains gates of at least $\sum_{K'\in \C'} m_{K'}/4c_1 c c_0 s\log^3 n$ row-classes each of 
cost at least $r^2/4800 n\ell$.
This contributes to the cost of $W$ by at least $\frac{r^2}{4800 n \ell} \cdot  \frac{nr}{2400 c_1 c c_0 s\log^3 n} = \Theta(r^3/\ell s \log^3 n)$. 
The theorem follows.




\end{document}